\RequirePackage{fix-cm}
\documentclass[smallextended]{svjour3}       

\smartqed  

\usepackage{graphicx}
\usepackage{amssymb}
\usepackage{amsmath,CJK,indentfirst,amsmath,amsfonts,amssymb,cite}
\usepackage[colorlinks,linkcolor=blue,anchorcolor=blue,citecolor=blue,urlcolor=blue]{hyperref}
\usepackage{amssymb}
\usepackage{array}
\usepackage{booktabs}
\usepackage[usenames,dvipsnames]{color}
\usepackage{amsmath}
\usepackage{longtable}
\usepackage{multirow}
\definecolor{mathcolor2}{rgb}{1,.2,1}
\usepackage{longtable,array}\usepackage{multirow,multicol}
\usepackage[justification=centering]{caption}
\usepackage [latin1]{inputenc}
\allowdisplaybreaks[4]

=15pt

\def\square{\vcenter{\vbox{\hrule\hbox{\vrule
     \vbox to 8.8pt{\hbox to 10pt{}\vfill}\vrule}\hrule}}}

\newcommand{\cC}{{\mathcal{C}}}

\newcommand{\tr}{{\mathrm{Tr}}}

\newcommand{\gf}{{\mathbb{F}}}

\newcommand{\support}{{\mathrm{suppt}}}

\newcommand{\C}{{\mathcal{C}}}

\newcommand{\bc}{{\mathbf{c}}}
\newcommand{\bg}{{\mathbf{g}}}

\begin{document}

\title{A construction of optimal locally recoverable codes}


\author{Xiaoru Li       \and   Ziling Heng
}


\institute{
X. Li and Z.Heng are with the School of Science, Chang'an University, Xi'an, 710064, China. \email{lixiaoru@163.com, zilingheng@chd.edu.cn}\\
Z. Heng's research was supported in part by the National Natural Science Foundation of China under Grant 11901049, in part by the  Young Talent Fund of University Association for Science and Technology in Shaanxi, China, under Grant 20200505 and in part by the Fundamental Research Funds for the Central Universities, CHD, under Grant 300102122202.}
\date{Received: date / Accepted: date}

\maketitle

\begin{abstract}
Locally recoverable codes are widely used in distributed and cloud storage systems. The objective of this paper is to present a construction of near MDS codes with oval polynomials and then determine the locality of the codes. It turns out that the near MDS codes and their duals are both distance-optimal and dimension-optimal locally recoverable codes. The lengths of the locally recoverable codes are different from known ones in the literature. 
\end{abstract}

\section{Introduction}
Let $\gf_q$ be the finite field with $q$ elements, where $q$ is a power of a prime $p$. Let $\mathbb{F}_{q}^*:=\mathbb{F}_{q}\setminus \{0\}$.

Let $\mathcal{C} \subseteq \mathbb{F}_{q}^{n}$ with $\C$ a non-empty set. Then $\mathcal{C}$ is said to be an $[n,k,d]$ \emph{linear code} over $\mathbb{F}_q$ if it is a $k$-dimensional linear subspace over $\mathbb{F}_{q}$, where $d$ denotes its minimal distance. Define the \emph{dual} of an $[n,k,d]$ linear code $\C$ by
$$
\mathcal{C}^{\perp}=\left\{ \textbf{u} \in \mathbb{F}_{q}^{n}: \langle  \textbf{u}, \mathbf{c} \rangle=0\mbox{ $\forall$ }\mathbf{c} \in \mathcal{C} \right\},
$$
where $\langle \textbf{u},\mathbf{c} \rangle$ denotes the Euclidean inner product of $\textbf{u}$ and $\mathbf{c}$. It is obvious that $\mathcal{C}^{\perp}$ is an $[n,n-k]$ linear code. Denote by $A_{i}$ the number of codewords of weight $i$ in an $[n,k]$ linear code $\mathcal{C}$ for $0 \leq i \leq n$. The polynomial $A(z)=1+A_{1}z+A_{2}z^2+ \cdots +A_{n}z^n$ is referred to as the \emph{weight enumerator} of $\mathcal{C}$. The weight enumerator is an interesting research project as it not only contains the error detection and error correction capabilities of the code, but also is useful for calculating the error probability of error detection.
In recent years, the weight enumerators of linear codes were widely studied in the literature \cite{D, DT, HDZ, L1, L2, TD, Wqiu}.

Linear codes achieving or nearly achieving the Singleton bound are important  in both theory and practice.
Linear codes of parameters $[n, k, n-k+1]$ are called MDS (maximum distance separable) codes.
An $[n, k, n-k]$ linear code is said to be almost maximum distance separable (AMDS for short).
A linear code is referred to as an near maximum distance separable (near MDS or NMDS for short) code provided that
both this code and its dual are AMDS.
Constructions of MDS, AMDS and NMDS as well as their applications were investigated in \cite{D, DT, DodLan95, HY, Shix1, Shix2, TanP, TD, Wqiu, W, WHL}.

Locally recoverable codes (LRCs for short) are widely used in distributed data storage systems.
A LRC with locality $r$ is a  block code such that any symbol in the encoding is a function of $r$ other symbols. In this letter, we only consider linear LRCs. Denote by $[n]=\{ 0,1,\cdots,n-1 \}$ with $n$ a positive integer. For an $[n,k,d]$ linear code $\mathcal{C}$ over $\mathbb{F}_q$, let the coordinates of the codewords in $\mathcal{C}$ be indexed by the elements in $[n]$. For any $i \in [n]$, if there always exist a subset $R_{i} \subseteq [n] \backslash {i}$ with $|R_i|=r$ and a function $f_{i}(x_1,x_2,\cdots,x_r)$ over $\mathbb{F}_q^{r}$ such that $c_i=f_{i}(\mathbf{c}_{R_i})$ for any $\mathbf{c}=(c_0,\cdots,c_{n-1}) \in \mathcal{C}$, then $\mathcal{C}$ is referred to as an $(n,k,d,q;r)$-LRC, where $\mathbf{c}_{R_i}$ is the projection of $\mathbf{c}$ at $R_{i}$ and the set $R_{i}$ is called the repair set of $c_i$. There exist some tradeoffs between the parameters of LRCs. For any $(n,k,d,q;r)$-LRC, the Singleton-like bound (see \cite{CM}) is given by
\begin{eqnarray}\label{eqn-slbound}
d \leq n-k- \left \lceil \frac{k}{r} \right \rceil +2.
\end{eqnarray}
LRCs are said to be distance-optimal if they achieve this bound.
For any $(n,k,d,q;r)$-LRC, the Cadambe-Mazumdar bound (see \cite{GH}) is given as
\begin{eqnarray}\label{eqn-cmbound}
k \leq \mathop{\min}_{t \in \Bbb Z^{+}} [rt+k_{opt}^{(q)}(n-t(r+1),d)],
\end{eqnarray}
where $k_{opt}^{(q)}(n,d)$ denotes the largest possible dimension of a linear code of length $n$, minimum distance $d$ over $\gf_q$, $\Bbb Z^{+}$ denotes the set of all positive integers.
LRCs are called dimension-optimal if they achieve the Cadambe-Mazumdar bound.
Constructing distance-optimal or dimension-optimal LRCs is an interesting research topic. In \cite{GY} and \cite{TanP}, AMDS and NMDS codes were used to derive optimal or nearly optimal locally recoverable codes. Hence it is interesting to construct new AMDS or NMDS codes with desired locality.

In \cite{Wqiu}, several families of  NMDS codes with some special matrixes were constructed. The objective of this paper is to present a construction of NMDS codes with larger lengths than those in \cite{Wqiu}  and then determine the locality of the codes. It turns out that the NMDS codes and their duals are both distance-optimal and dimension-optimal LRCs.

The rest of this paper is organized as follows. In Section \ref{sec2}, we present some preliminaries on NMDS codes and oval polynomials. In Section \ref{sec3}, we construct a family of $[q+5,3,q+2]$ NMDS codes with $q=2^m$. In Section \ref{sec4}, the localities of the NMDS codes and their duals are determined.
In Section \ref{sec4}, we conclude this paper and give some remarks.
\section{Preliminaries}\label{sec2}
Let $(1,A_1,\cdots,A_n)$ and $(1,A_1^\perp,\cdots,A_n^\perp)$  respectively denote the weight distributions of a linear code $\C$ and its dual $\C^\perp$ with length $n$. The weight distributions of an NMDS code and its dual satisfy the following recurrence relations.

\begin{lemma}[\cite{DodLan95}]\label{lem-N1}
Let $\mathcal{C}$ be an $[n, k]$ NMDS code over $\gf_q$. Then
\begin{eqnarray*}
A_{k+s}^\perp = \binom{n}{k+s} \sum_{j=0}^{s-1} (-1)^j \binom{k+s}{j}(q^{s-j}-1) +
             (-1)^s \binom{n-k}{s}A_{k}^\perp
\end{eqnarray*}
for $s \in \{1,2, \ldots, n-k\}$; and
\begin{eqnarray*}
A_{n-k+s} = \binom{n}{k-s} \sum_{j=0}^{s-1} (-1)^j \binom{n-k+s}{j}(q^{s-j}-1)
+(-1)^s \binom{k}{s}A_{n-k}
\end{eqnarray*}
for $s \in \{1,2, \ldots, k\}$.
\end{lemma}

The following theorem shows an interesting property of NMDS codes.

\begin{lemma}[\cite{FaldumWillems97}]\label{lem-N2}
Let $\mathcal{C}$ be an NMDS code and $\support(\bc)=\{1 \leq i \leq n: c_i \neq 0\}$
denote the support of $\bc=(c_1, \ldots, c_n)\in \C$.
Then for any minimum weight codeword $\bc$ in $\mathcal{C}$, there exists, up to a multiple, a unique minimum weight codeword $\bc^\perp$ in $\mathcal{C}^\perp$ satisfying
$\support(\bc) \cap \support(\bc^\perp)=\emptyset$. Besides, $\mathcal{C}$ and $\C^\perp$ have the same number of minimum weight codewords.
\end{lemma}

Next we list the definition and some properties of oval polynomial used in this letter.

\begin{definition} \cite{LN97} Let $q=2^m$ with $m\geq 2$.
If $f \in \gf_q[x]$ is a polynomial satisfying
\begin{enumerate}
\item $f$ is a permutation polynomial of $\gf_q$ with $\deg(f)<q$ and $f(0)=0$, $f(1)=1$;
\item for each $a \in \gf_q$, $g_a(x):=(f(x+a)+f(a))x^{q-2}$ is also a permutation polynomial of $\gf_q$,
\end{enumerate}
then $f$ is called an oval polynomial.
\end{definition}

To construct near MDS codes over $\gf_q$ in the paper, we need concrete oval polynomials over $\gf_q$. Some known infinite families of oval polynomials are listed in the following.

\begin{theorem}\label{thm-knownopolys}\cite[Table 1]{M}
Let $m \geq 2$ be an integer. The following are oval polynomials of $\gf_q$, where $q=2^m$.
\begin{itemize}
\item The translation polynomial $f(x)=x^{2^h}$, where $\gcd(h, m)=1$.
\item The Segre polynomial $f(x)=x^6$, where $m$ is odd.
\item The Glynn oval polynomial $f(x)=x^{3 \times 2^{(m+1)/2} +4}$, where $m$ is odd.
\item The Glynn oval polynomial $f(x)=x^{ 2^{(m+1)/2} + 2^{(m+1)/4} }$ for $m \equiv 3 \pmod{4}$.
\item The Glynn oval polynomial $f(x)=x^{ 2^{(m+1)/2} + 2^{(3m+1)/4} }$ for $m \equiv 1 \pmod{4}$.
\item The Cherowitzo oval polynomial $f(x)=x^{2^e}+x^{2^e+2}+x^{3 \times 2^e+4},$ where $e=(m+1)/2$ and $m$ is odd.
\item The Payne oval polynomial $f(x)=x^{\frac{2^{m-1}+2}{3}} + x^{2^{m-1}} + x^{\frac{3 \times 2^{m-1}-2}{3}}$,
        where $m$ is odd.
\item The Subiaco polynomial
$$
f_a(x)=((a^2(x^4+x)+a^2(1+a+a^2)(x^3+x^2)) (x^4 + a^2 x^2+1)^{2^m-2}+x^{2^{m-1}},
$$
where $\tr_{q/2}(1/a)=1$ and $a \not\in \gf_4$ if $m \equiv 2 \bmod{4}$.
\item The Adelaide oval polynomial
$$
f(x)=\frac{T(\beta^m)(x+1)}{T(\beta)} + \frac{T((\beta x + \beta^q)^m)}{T(\beta) (x+T(\beta)x^{2^{m-1}} +1)^{m-1}} + x^{2^{m-1}},
$$
where $m \geq 4$ is even, $\beta \in \gf_{q^2} \setminus \{1\}$ with $\beta^{q+1}=1$, $m \equiv \pm (q-1)/3 \pmod{q+1}$,
and $T(x)=x+x^q$.
\end{itemize}
\end{theorem}

\begin{lemma}\cite{Masch98}\label{lem-opoly1}
Let $q=2^m$ with $m\geq 2$. Then a polynomial $f$ with $f(0)=0$ over $\gf_q$ is an oval polynomial if and only if $f_u:=f(x)+ux$ is $2$-to-$1$ for each $u \in \gf_q^*$.
\end{lemma}

\begin{lemma}\cite{Wqiu}\label{lem-opoly2}
Let $q=2^m$ with $m\geq 2$. Then $f$ is an oval polynomial over $\gf_q$ if and only if the following two conditions hold:
\begin{enumerate}
\item $f$ is a permutation polynomial of $\gf_q$;
\item
$$
\frac{f(x)+f(y)}{x+y} \neq \frac{f(x)+f(z)}{x+z}
$$
for all pairwise different elements $x, y, z$ in $\gf_q$.
\end{enumerate}
\end{lemma}

\begin{lemma}\cite{Wqiu}\label{lem-opoly3}
Let $q=2^m$ with $m \geq 3$ being odd and $f(x)$ be an oval polynomial over $\gf_q$ whose coefficients are in $\gf_2$. Then $f(x)+x+1=0$ has no solution in $\gf_q$.
\end{lemma}

\section{A Construction of NMDS codes}\label{sec3}
In this section, let $q=2^m$ where $m$ is an odd integer with $m \geq 3$. For convenience, let $\dim(\cC)$ and $d(\cC)$ respectively denote the dimension and minimal distance of a linear code $\cC$. Let $f$ be an oval polynomial over $\gf_q$. Let $\alpha_0=0,\alpha_1=1,\cdots,\alpha_{q-1}$ be all elements of $\gf_q$.
Define
\begin{eqnarray*}
G=\left[
\begin{array}{lllllllll}
1 & 1 & \cdots & 1  & 0 & 0 & 1 & 0 & 1\\
\alpha_0 & \alpha_1 & \cdots & \alpha_{q-1} & 0 & 1 & 0 & 1 & 1\\
f(\alpha_0) & f(\alpha_1) & \cdots & f(\alpha_{q-1}) & 1 & 0 & 1 & 1 &0
\end{array}
\right].
\end{eqnarray*}
Then $G$ is a $3$ by $q+5$ matrix over $\gf_q$. Let $G$ generate a linear code $\C$ over $\gf_q$.
The parameters and weight enumerator of $\C$ are determined in the following theorem.

\begin{theorem}\label{thm-N}
Let $m$ be an odd integer with $m \geq 3$, and let $f$ be an oval polynomial over $\gf_q$. Then $\cC$ is a $[q+5, 3, q+2]$ NMDS code over $\gf_q$ with  weight enumerator
\begin{eqnarray*}
A(z)=1 + \frac{(q-1)(3q+8)}{2} z^{q+2} + \frac{(q-2)(q-1)(q+2)}{2} z^{q+3} + \\
\frac{3(q^2-3q+2)}{2} z^{q+4} + \frac{(q-1)(q-2)^2}{2} z^{q+5}.
\end{eqnarray*}
\end{theorem}

\begin{proof}
It is easy to deduce that $\dim(\cC)=3$ as the first, $q+1$-th and $q+2$-th columns of the generator matrix $G$ are linearly independent. We then prove that $\cC^\perp$ has parameters $[q+5, q+2, 3]$. Obviously, $\dim(\cC^\perp)=(q+5)-3=q+2$. Since each column of $G$ is nonzero and any two columns of $G$ are linearly independent over $\gf_q$, we have $d(\cC^\perp) > 2$. Note that the $q+1$-th, $q+2$-th, $q+4$-th columns of $G$ are linearly dependent. Then  $d(\cC^\perp) = 3$.

To calculate the total number of codewords of weight $3$ in $\cC^\perp$, we consider the following cases.

{Case 1.1:} Let $x, y, z$ be three pairwise different elements in $\gf_q$. Consider the submatrix
\begin{eqnarray*}
M_{1,1}=\left[
\begin{array}{lll}
1 & 1 & 1 \\
x & y & z \\
f(x) & f(y) & f(z)
\end{array}
\right].
\end{eqnarray*}
 Then $|M_{1,1}|=(x+y)(f(x)+f(z))+(x+z)(f(x)+f(y))\neq 0$ by Lemma \ref{lem-opoly2}. Hence, $\cC^\perp$ has no codeword of weight $3$ whose nonzero coordinates are at the first $q$ locations.

{Case 1.2:} Let $x, y$ be two different elements in $\gf_q$. Consider the submatrix
\begin{eqnarray*}
M_{1,2}=\left[
\begin{array}{lll}
1 & 1 & 0 \\
x & y & 0 \\
f(x) & f(y) & 1
\end{array}
\right].
\end{eqnarray*}
Then $|M_{1,2}|=y+x \neq 0$  as $x \neq y$.  Hence, $\cC^\perp$ has no codeword of weight $3$ whose first two nonzero coordinates are at the first $q$ locations and the rest is at the $q+1$-th location.

{Case 1.3:} Let $x, y$ be two different elements in $\gf_q$. Consider the submatrix
\begin{eqnarray*}
M_{1,3}=\left[
\begin{array}{lll}
1 & 1 & 0 \\
x & y & 1 \\
f(x) & f(y) & 0
\end{array}
\right].
\end{eqnarray*}
Since $f$ is a permutation polynomial and $x \neq y$, then $|M_{1,3}|=f(y)+f(x) \neq 0$. Hence, $\cC^\perp$ has no codeword of weight $3$ whose first two nonzero coordinates are at the first $q$ locations and the rest is at the $q+2$-th location.

{Case 1.4:} Let $x, y$ be two different elements in $\gf_q$. Consider the submatrix
\begin{eqnarray*}
M_{1,4}=\left[
\begin{array}{lll}
1 & 1 & 1 \\
x & y & 0 \\
f(x) & f(y) & 1
\end{array}
\right].
\end{eqnarray*}
Then $|M_{1,4}|=(f(y)+1)x+(f(x)+1)y$. If $(x,y)=(0,1)$ or $(1,0)$, then $|M_{1,4}| \neq 0$ and $\cC$ has no codeword of weight $3$ whose coordinates are at the first, second and $q+3$-th locations. Now we count the number of different pairs $(x,y)$ such that $|M_{1,4}|=0$, where $x,y \in \gf_q \setminus \{0,1\}$ . For any $x \in \gf_q \setminus \{0,1\}$, $|M_{1,4}|=0$ if and only if
$$\frac{f(x)+1}{x}=\frac{f(y)+1}{y}.$$
Let $a:=\frac{f(x)+1}{x}$. Then $a \neq 0$ and $a \neq 1$ by Lemma \ref{lem-opoly3}. Since $f(z)+az$ is $2$-to-$1$ by Lemma \ref{lem-opoly1}, there exists an unique element $y \in \gf_q \setminus \{0,1\}$ such that $f(x)+ax=1=f(y)+ay$. For this pair $(x, y)$, $|M_{1,4}|=0$ and vice versa. Hence, the number of distinct $(x,y) \in \gf_q \setminus \{0,1\}$ such that $|M_{1,4}|=0$ equals $(q-2)/2$. As a result, the number of codewords of weight $3$ in $\cC^\perp$ whose first two nonzero coordinates are at the first $q$ locations (expect the first two locations) and the rest is at the $q+3$-th location is equal to $(q-2)(q-1)/2$.

{Case 1.5:} Let $x, y$ be two distinct elements in $\gf_q$. Consider the submatrix
\begin{eqnarray*}
M_{1,5}=\left[
\begin{array}{lll}
1 & 1 & 0 \\
x & y & 1 \\
f(x) & f(y) & 1
\end{array}
\right].
\end{eqnarray*}
Then $|M_{1,5}|=f(x)+f(y)+x+y$. $|M_{1,5}|=0$ is equal to $f(x)+x=f(y)+y$. Note that $f(z)+z$ is $2$-to-$1$ by Lemma \ref{lem-opoly1}. If we fix $x \in \gf_q$, there exists an unique element $y \in \gf_q$ such that $f(x)+x=a=f(y)+y$, where $a \in \gf_q$. For this pair $(x, y)$, $|M_{1,5}|=0$ and vice versa. Then  the number of $(x,y)$ in $\gf_q$ such that $|M_{1,5}|=0$ equals $q/2$. In conclusion, the number of codewords of weight $3$ in $\cC^\perp$ whose first two nonzero coordinates are at the first $q$ locations and the rest is at the $q+4$-th location is equal to $q(q-1)/2$.

{Case 1.6:} Let $x, y$ be two different elements in $\gf_q$. Consider the submatrix
\begin{eqnarray*}
M_{1,6}=\left[
\begin{array}{lll}
1 & 1 & 1 \\
x & y & 1 \\
f(x) & f(y) & 0
\end{array}
\right].
\end{eqnarray*}
Then $|M_{1,6}|=(x+1)f(y)+(y+1)f(x)$. For any $y \in \gf_q \setminus \{0,1\}$, let $a:=f(y)/(y+1)$ which implies $f(y)+ay=a$. Then $a \neq 0$ and $a \neq 1$ by Lemma \ref{lem-opoly3}. By Lemma \ref{lem-opoly1}, $f(z)+az$ is $2$-to-$1$. Then there exists an unique element $x \in \gf_q \setminus \{0,1\}$ such that $f(x)+ax=a$. For this pair $(x, y)$, $|M_{1,6}|= 0$ and and vice versa. Hence, the number of  $(x, y)$ in $\gf_q \setminus \{0,1\}$ satisfying $|M_{1,6}|=0$ equals $(q-2)/2$. Consequently, the number of codewords of weight $3$ in $\cC^\perp$ whose first two nonzero coordinates are at the first $q$ locations (expect the first two locations) and the rest is at the $q+5$-th location is equal to $(q-2)(q-1)/2$.

{Case 1.7:} Let $x$ be an element in $\gf_q$. Consider the following three submatrixes as
\begin{eqnarray*}
M_{1,7}=\left[
\begin{array}{lll}
1 & 0 & 0 \\
x & 0 & 1 \\
f(x) & 1 & 0
\end{array}
\right],\
M_{1,8}=\left[
\begin{array}{lll}
1 & 0 & 0 \\
x & 0 & 1 \\
f(x) & 1 & 1
\end{array}
\right],
\end{eqnarray*}
\begin{eqnarray*}
M_{1,9}=\left[
\begin{array}{lll}
1 & 0 & 0 \\
x & 1 & 1 \\
f(x) & 0 & 1
\end{array}
\right].
\end{eqnarray*}
Then we have $|M_{1,7}|=|M_{1,8}|=|M_{1,9}|=1$.  Hence, $\cC^\perp$ has no codeword of weight $3$ whose first nonzero coordinate is at the first $q$ locations and the others are at the $u$-th and $v$-th locations, where $(u,v)=(q+1,q+2)$ or $(q+1,q+4)$ or $(q+2,q+4)$.

{Case 1.8:} Let $x$ be an element in $\gf_q$. Consider the submatrix
\begin{eqnarray*}
M_{1,10}=\left[
\begin{array}{lll}
1 & 0 & 1 \\
x & 0 & 0 \\
f(x) & 1 & 1
\end{array}
\right].
\end{eqnarray*}
Then $|M_{1,10}|=x$. $|M_{1,10}|=0$ if and only if $x=0$. Consequently, the number of codewords of weight $3$ in $\cC^\perp$ whose nonzero coordinates are at the first, $q+1$-th and $q+3$-th locations is equal to $q-1$.

{Case 1.9:} Let $x$ be an element in $\gf_q$. Consider the submatrix
\begin{eqnarray*}
M_{1,11}=\left[
\begin{array}{lll}
1 & 0 & 1 \\
x & 0 & 1 \\
f(x) & 1 & 0
\end{array}
\right].
\end{eqnarray*}
Then $|M_{1,11}|=x+1$. $|M_{1,11}|=0$ if and only if $x=1$. Consequently, the number of codewords of weight $3$ in $\cC^\perp$ whose nonzero coordinates are at the second, $q+1$-th and $q+5$-th locations equals $q-1$.

{Case 1.10:} Let $x$ be an element in $\gf_q$. Consider the submatrix
\begin{eqnarray*}
M_{1,12}=\left[
\begin{array}{lll}
1 & 0 & 1 \\
x & 1 & 0 \\
f(x) & 0 & 1
\end{array}
\right].
\end{eqnarray*}
Then $|M_{1,12}|=f(x)+1$. Since $f$ is a permutation polynomial of $\gf_q$ with $f(0)=0$, $f(1)=1$, then $|M_{1,12}|=0$ if and only if $x=1$. Consequently, the number of codewords of weight $3$ in $\cC^\perp$ whose nonzero coordinates are at the second, $q+2$-th and $q+3$-th locations equals $q-1$.

{Case 1.11:} Let $x$ be an element in $\gf_q$. Consider  the submatrix
\begin{eqnarray*}
M_{1,13}=\left[
\begin{array}{lll}
1 & 0 & 1 \\
x & 1 & 1 \\
f(x) & 0 & 0
\end{array}
\right].
\end{eqnarray*}
Then $|M_{1,13}|=f(x)$. Since $f$ is a permutation polynomial of $\gf_q$ with $f(0)=0$, $f(1)=1$, then $|M_{1,13}|=0$ if and only if $x=0$. Consequently, the number of codewords of weight $3$ in $\cC^\perp$ whose nonzero coordinates are at the first, $q+1$-th and $q+5$-th locations is equal to $q-1$.

{Case 1.12:} Let $x$ be an element in $\gf_q$. Consider the following three submatrixes as
\begin{eqnarray*}
M_{1,14}=\left[
\begin{array}{lll}
1 & 1 & 0 \\
x & 0 & 1 \\
f(x) & 1 & 1
\end{array}
\right],
M_{1,15}=\left[
\begin{array}{lll}
1 & 1 & 1 \\
x & 0 & 1 \\
f(x) & 1 & 0
\end{array}
\right],
\end{eqnarray*}
\begin{eqnarray*}
M_{1,16}=\left[
\begin{array}{lll}
1 & 1 & 0 \\
x & 1 & 1 \\
f(x) & 0 & 1
\end{array}
\right].
\end{eqnarray*}
Then we have $|M_{1,14}|=|M_{1,15}|=|M_{1,16}|=f(x)+x+1$. By Lemma \ref{lem-opoly3}, $f(x)+x+1\neq 0$ for $x\in \gf_q$. Hence, $\cC^\perp$ has no codeword of weight $3$ whose first nonzero coordinate is at the first $q$ locations and the others are at the $t_1$-th and $t_2$-th locations, where
$(t_1,t_2)=(q+3,q+4)$ or $(q+3,q+5)$ or $(q+4,q+5)$.

{Case 1.13:} Let the nonzero coordinates of the codewords with weight $3$ in $\cC^\perp$ be at three of the last five locations.
Then there are ten subcases.
Here, we only discuss two subcases of them as the others can be discussed in a similar way.
Consider the following two submatrixes as
\begin{eqnarray*}
M_{1,17}=\left[
\begin{array}{lll}
0 & 0 & 0 \\
0 & 1 & 1 \\
1 & 0 & 1
\end{array}
\right],
M_{1,18}=\left[
\begin{array}{lll}
1 & 0 & 1 \\
0 & 1 & 1 \\
1 & 1 & 0
\end{array}
\right].
\end{eqnarray*}
The ranks of these submatrixes are all $2$. Hence, the number of codewords of weight $3$ in $\cC^\perp$ whose nonzero coordinates are at the $q+1$-th, $q+2$-th and $q+4$-th locations is equal to $q-1$ and the number of codewords of weight $3$ in $\cC^\perp$ whose nonzero coordinates are at the last three locations is equal to $q-1$.
For other subcases, the corresponding number of codewords of weight $3$ in $\cC^\perp$  is 0.

Thanks to the discussions in the above cases, we deduce that $A_3^\perp=\frac{(q-1)(3q+8)}{2}$.

We finally prove that $d(\cC)=q+2$. Assume that $d(\cC) \leq q+1=q+5-4$. Let $\bc=a \bg_1 + b \bg_2 + c \bg_3$ be a codeword with the minimum weight in $\cC$, where $\bg_1$, $\bg_2$ and $\bg_3$ respectively represent the first, second and third rows of $G$. Then at least four coordinates are zero in $\bc$.

{Case 2.1:} Assume that four of the last five coordinates in $\bc$ are zero. Here we suppose the $q+1$-th, $q+2$-th, $q+3$-th, $q+4$-th coordinates in $\bc$ are zero. Then we have
\begin{eqnarray*}
\left\{
\begin{array}{r}
c=0, \\
b=0, \\
a+c=0, \\
b+c=0.  \\
\end{array}
\right.
\end{eqnarray*}
Then $a=b=c=0$ and $\bc=0$, which contradicts to the fact that $\bc$ is a minimum weight codeword in $\cC$.
For other subcases, we can similarly obtain this contradiction.

{Case 2.2:} Assume that three of the last five coordinates in $\bc$ are zero. Here we suppose the $q+1$-th, $q+2$-th, $q+3$-th  coordinates in $\bc$ are zero. Then there exists an element $x$ in $\gf_q$ such that
\begin{eqnarray*}
\left\{
\begin{array}{r}
a+bx+cf(x) = 0, \\
c=0, \\
b=0, \\
a+c=0.  \\
\end{array}
\right.
\end{eqnarray*}
Then $a=b=c=0$ and $\bc=0$. This is contrary to the fact that $\bc$ is a minimum weight codeword in $\cC$. For other subcases, we can similarly obtain this contradiction.

{Case 2.3:} Assume that two of the last five coordinates in $\bc$ are zero. Here we suppose the $q+1$-th, $q+2$-th of the last five coordinates in $\bc$ are zero. Then there exist two different elements $x$ and $y$ in $\gf_q$ such that
\begin{eqnarray*}
\left\{
\begin{array}{r}
a+bx+cf(x) = 0, \\
a+by+cf(y) = 0, \\
c=0,  \\
b=0. \\
\end{array}
\right.
\end{eqnarray*}
Then $a=b=c=0$ and $\bc=0$. This is contrary to the fact that $\bc$ is a minimum weight codeword in $\cC$. For other subcases, we can similarly obtain this contradiction.

{Case 2.4:} Assume that at most one of the last four coordinates in $\bc$ is zero. Let $x, y, z$ be three pairwise different elements in $\gf_q$ such that
\begin{eqnarray*}
\left\{
\begin{array}{r}
a+bx+cf(x) = 0, \\
a+by+cf(y) = 0, \\
a+bz+cf(z) = 0. \\
\end{array}
\right.
\end{eqnarray*}
By Lemma \ref{lem-opoly2},  we can deduce that the rank of the coefficient matrix for this system of equations is $3$. Hence, $a=b=c=0$ and $\bc=0$, which is contrary to the fact that $\bc$ is a minimum weight codeword in $\cC$.

Summarizing the above discussions, $d(\cC) \geq q+2$. By the Singleton bound, $d(\cC) \leq q+3$. If $d(\cC) = q+3$, then $\cC$ is a $[q+5, 3, q+3]$ MDS code whose dual is also an MDS code, which contradicts to the fact that $\C^\perp$ is AMDS. Then $d(\cC)=q+2$, and $\cC$ is a $[q+5, 3, q+2]$ NMDS code. By Lemma \ref{lem-N2}, $A_{q+2}=A_3^\perp=\frac{(q-1)(3q+8)}{2}$. Then the weight enumerator of $\cC$ follows form Lemma \ref{lem-N1}.
\end{proof}

\begin{example}
Let $m=3$ and $f(x)=x^4$. Then the code $\cC$ over $\gf_q$ has parameters $[13,3,10]$ and weight enumerator $A(z)=1 + 112z^{10} + 210z^{11} + 63z^{12} +126z^{13}.$
\end{example}

With the first seven families of oval polynomials documented in Theorem \ref{thm-knownopolys},  we have derived  seven infinite
families of near MDS codes over $\gf(q)$ with parameters $[q+5, 3, q+2]$ via Theorem \ref{thm-N}.
This construction may not work for the Subiaco and Adelaide oval polynomials in general.

\section{Optimal locally recoverable codes}\label{sec4}
In this section, we prove that the NMDS code  in Theorem \ref{thm-N} and its dual are both distance-optimal and dimension-optimal LRCs.

For an $[n,k,d]$ linear code $\C$, denote by $\mathcal{B}_{d}(\mathcal{C})$  the set of the supports of all codewords with weight $d$ in  $\mathcal{C}$.
Denote by $d^{\perp}=d(\mathcal{C}^{\perp})$.
Besides, the coordinates of the codewords are indexed with $(0,1,\cdots,n-1)$.

\begin{lemma}[\cite{TanP}]\label{lem-R1114}
Let $\cC$ be a linear code with length $n$ and $d(\C)>1$. Then the minimum linear locality of $\cC$ equals $d^{\perp}-1$ if and only if
$$
\mathop{\bigcup}_{\mathcal{S} \in \mathcal{B}_{d^{\perp}}(\mathcal{C}^{\perp})} \mathcal{S}=[n].
$$
\end{lemma}

\begin{lemma}[\cite{TanP}]\label{lem-R1115}
Let $\mathcal{C}$ be an NMDS code with $d(\C)>1$, then the minimum linear locality of $\mathcal{C}$ is either $d(\mathcal{C}^{\perp})-1$ or $d(\mathcal{C}^{\perp})$.
\end{lemma}

\begin{lemma}[\cite{TanP}]\label{lem-R1116}
Let $\mathcal{C}$ be an NMDS code. If
$$
\mathop{\bigcap}_{\mathcal{S} \in \mathcal{B}_{d^{\perp}}(\mathcal{C}^{\perp})} \mathcal{S}=\emptyset,
$$
then the minimum linear locality of $\mathcal{C}^{\perp}$ is equal to $d(\mathcal{C})-1$.
\end{lemma}

\begin{theorem}\label{th-main}
The NMDS code $\cC$ in Theorem \ref{thm-N} is a
$$
(q+5,3,q+2,q;2)-\mbox{LRC}
$$
and $\cC^{\perp}$ is a
$$
(q+5,q+2,3,q;q+1)-\mbox{LRC}.
$$
In addition, $\cC$ and $\cC^{\perp}$ are both distance-optimal and dimension-optimal LRCs.
\end{theorem}

\begin{proof}
By the proof of Theorem \ref{thm-N}, it is easy to deduce that
$$\mathop{\bigcup}_{\mathcal{S} \in \mathcal{B}_{3}(\mathcal{C}^{\perp})} \mathcal{S}=[q+5]$$ and $$\mathop{\bigcap}_{\mathcal{S} \in \mathcal{B}_{3}(\mathcal{C}^{\perp})} \mathcal{S}=\emptyset.$$
Then by Lemma \ref{lem-R1114}, the minimum linear locality of $\cC$ is $d(\cC^{\perp})-1=2$.
By Theorem \ref{lem-R1116}, the minimum linear locality of $\cC^{\perp}$ is $d(\cC)-1=q$.
Now we prove $\cC$ is an optimal LRC. Putting the parameters of the $(q+5,3,q+2,q;2)$-LRC into the right-hand side of the Singleton-like bound in (\ref{eqn-slbound}), we have
$$
n-k- \left \lceil \frac{k}{r} \right \rceil +2
=q+5-3- \left \lceil \frac{3}{2} \right \rceil +2=q+2.
$$
Hence $\cC$ is a distance-optimal LRC.
Putting $t=1$ and the parameters of the $(q+5,3,q+2,q;2)$-LRC into the right-hand side of the Cadambe-Mazumdar bound in (\ref{eqn-cmbound}), we have
$$
k \leq r+k_{opt}^{(q)}(n-(r+1),d)
=2+k_{opt}^{(q)}(q+2,q+2) = 3,
$$
where  $k_{opt}^{(q)}(q+2,q+2) = 1$ by the classical Singleton bound. Thus, $\cC$ is a dimension-optimal LRC. Similarly, we can prove $\cC^{\perp}$ is both distance-optimal and dimension-optimal.
\end{proof}

\section{Concluding remarks}\label{sec5}
With the special generator matrix $G$ and an oval polynomial $f(x)$, we presented a construction of  $[q+5,3,q+2]$ NMDS code $\C$ in Theorem \ref{thm-N} for $q=2^m$ and odd $m$.
Then we derived seven  infinite families of $[q+5,3,q+2]$ NMDS codes with the first seven families of oval polynomials documented in Theorem \ref{thm-knownopolys}.
The NMDS codes $\C$ and $\C^\perp$ were proved to be both distance-optimal and dimension-optimal LRCs in Theorem \ref{th-main}.

In \cite{LC}, a class of optimal locally repairable codes of distances 3 and 4 with unbounded length was constructed.
We remark that the optimal locally repairable codes  in this paper are not contained in \cite{LC} as they have different lengths.
Finally, we point out that the NMDS codes in this paper have larger lengths than those in \cite{Wqiu}.

\end{document}